\documentclass[namedate,webpdf,modern,mediumone]{oup-authoring-template}

\usepackage{amsmath,amssymb}
\usepackage{graphicx}
\usepackage{booktabs}

\newcommand{\societylogo}{}

\graphicspath{{./}}

\DeclareMathOperator*{\argmax}{arg\,max}
\DeclareMathOperator{\var}{var}
\DeclareMathOperator{\sd}{sd}
\DeclareMathOperator{\tr}{tr}
\newcommand{\adist}{\mathrel{\dot\sim}}
\newcommand{\bb}{\boldsymbol{\beta}}
\newcommand{\bbh}{\hat{\boldsymbol{\beta}}}
\newcommand{\bbt}{\tilde{\boldsymbol{\beta}}}
\newcommand{\by}{\boldsymbol{y}}
\newcommand{\bpi}{\boldsymbol{\pi}}

\theoremstyle{thmstyleone}
\newtheorem{proposition}{Proposition}
\theoremstyle{thmstyletwo}
\newtheorem{remark}{Remark}

\begin{document}

\journaltitle{}
\DOI{}
\copyrightyear{}
\pubyear{2026}
\vol{}
\issue{}
\access{}
\appnotes{Preprint}
\firstpage{1}

\title[Hosmer--Lemeshow after shrinkage]{Shrinkage invalidates the Hosmer--Lemeshow
test: goodness of fit for penalized logistic regression, with an application to glaucoma
diagnosis}

\author[1,$\ast$]{Ebrahim Khaled Ebrahim}

\address[1]{\orgdiv{Department of Applied Statistics},
\orgname{Alexandria University}, \orgaddress{\country{Egypt}}}

\corresp[$\ast$]{Corresponding author. \href{mailto:ebrahimkhaled@alexu.edu.eg}
{ebrahimkhaled@alexu.edu.eg}}

\abstract{Clinical prediction models are increasingly fitted by penalized logistic
regression, because collinearity or many candidate predictors makes maximum likelihood
unstable or impossible. Calibration is then almost always assessed by a grouped
goodness-of-fit test such as the Hosmer--Lemeshow test. We show that this combination is
invalid. Under ridge regression the grouped standardized residuals acquire a
non-centrality induced by shrinkage, so the reference distribution used in practice is
wrong, and at the penalty that most improves the fitted probabilities the test rejects
correctly specified models between 92 and 100 per cent of the time. We derive the
corrected law and define the shrinkage-corrected Hosmer--Lemeshow test, which subtracts
an estimate of that non-centrality, restoring the maximum likelihood reference exactly to
first order, and is made valid by prepivoting at a power cost we measure. We also give the attenuation law governing what
any such test can detect once the linear predictor must be estimated. In glaucoma
diagnosis by confocal laser tomography, where the maximum likelihood estimate does not
exist, the corrected test finds the evidence for misfit weaker by more than three orders
of magnitude: the fitted risks are too flat rather than mis-ordered, so the model needs
recalibration rather than rebuilding.}


\maketitle

\section{Introduction}\label{sec:intro}

Glaucoma is a leading cause of \emph{irreversible} blindness worldwide, and it is
insidious: the disease usually progresses without symptoms until a substantial part of
the visual field has already been lost, and what has been lost cannot be recovered
\citep{coan2023}. Everything therefore depends on detecting the disease early, while
treatment can still preserve the sight that remains. Diagnosis rests largely on the
appearance of the optic nerve head, and expert assessment of it is subjective, slow and
expensive \citep{coan2023}, which is why quantitative imaging and statistical
classification have been pursued for three decades.

We consider a diagnostic model of exactly this kind, built on confocal laser scanning
tomography of the optic nerve head. The data comprise $n = 196$ eyes, 98 with glaucoma
and 98 normal, each summarized by $p = 62$ morphometric variables: areas, volumes,
depths and slopes of the optic disc, measured globally and within six sectors, together
with height variation and cup shape measures. The clinical question is the ordinary one.
Given these measurements for a new eye, what is the probability that it is glaucomatous,
and can that probability be trusted enough to act on?

Anyone who develops clinical prediction models will recognize the sequence that follows,
because it is the one the reporting guidelines \citep{collins2015} and the sample-size
literature \citep{riley2021} prescribe:
penalize when the number of candidate predictors is large relative to the number of
events, tune the penalty by cross-validation, then report calibration, almost always with
a calibration plot and a Hosmer--Lemeshow test. This paper is about that sequence. Its
last step is not valid after its second, so a model reported as badly calibrated in a
development paper may have been condemned for using the estimator that the same guidance
recommended.

\subsection{Non-existence of the maximum likelihood estimate}\label{sec:cannotfit}

The variables are repeated geometric summaries of the same optic disc, so they are
collinear by construction rather than by accident: the largest pairwise correlation is
$0.996$ and the condition number of the cross-product matrix of the standardized
predictors is $5.0 \times 10^{7}$.
With $p/n = 0.32$ and collinearity of this order, the logistic likelihood is maximized
on a boundary; the iteratively reweighted least squares algorithm does not converge, and
the maximum likelihood estimate does not exist.

This is worth stating plainly because it removes a choice that is usually available. The
analyst cannot fit this model by maximum likelihood and then decide whether penalization
would be an improvement. Penalization is the only route to a fitted model at all. Ridge
estimation for logistic regression, in the form we use, is due to
\citet{lecessie1992}; with the penalty chosen by ten-fold cross-validation it gives a model that
discriminates well, with an area under the receiver operating characteristic curve of
$0.905$, and which a clinician might reasonably consider using.

\subsection{Assessment by the Hosmer--Lemeshow test}\label{sec:routinecheck}

Having fitted the model, the analyst asks whether it fits. In clinical prediction the
near-universal answer is a grouped goodness-of-fit test, and in practice almost always
the Hosmer--Lemeshow test \citep{hosmer1980,copas1989}, whose standing rests largely on
the comparison study of \citet{hosmer1997}:
sort the patients by predicted risk, cut them into ten groups, and compare the number of
events observed in each group with the number the model expected. For these data that test
returns $p < 10^{-5}$ at the penalty cross-validation selects, and it rejects at the
$5\%$ level in every one of the twelve penalty-and-grouping configurations we examined,
when the grouped statistic is referred to the maximum likelihood covariance --- the
sharper of the two forms in use, and the one we compare against throughout
(Section~\ref{sec:corrected}). The textbook $\chi^2_{G-2}$ form rejects in eleven of the
twelve.

Read in the ordinary way, that is a verdict: the model is badly misspecified, its
predicted probabilities cannot be believed, and it should be rebuilt before anyone
relies on it.

We show in this paper that the verdict is not available, because the test that produced
it is not valid for the model that was fitted. Its reference distribution is derived
under maximum likelihood estimation --- the very thing these data do not admit --- and
under penalized estimation it is wrong in a specific and quantifiable way. At the
penalty that cross-validation selects, that test rejects models we know to be correct
with probability approaching one.

\subsection{Relevance beyond a single dataset}\label{sec:whymatters}

Penalized regression is recommended whenever the number of candidate predictors is large
relative to the number of events \citep{riley2021,pavlou2024}, and the resulting model is
then assessed for calibration in the usual way. Calibration --- the agreement between predicted
probabilities and observed event rates --- deserves that attention, because a model that
discriminates well but is miscalibrated can still be misleading, and can be actively
harmful when used to guide decisions \citep{vancalster2019}. Miscalibration always
reduces the net benefit of using a model, and for some risk thresholds it makes a model
worse than treating everyone or no one \citep{vancalster2015}. Across $158$ external validations of $104$ cardiovascular prediction models, $91\%$
carried a risk of harm at some plausible threshold \citep{gulati2022}.

So the practice is right to check calibration. Our claim is narrower and, we think, more
awkward: the standard check does not do what it is believed to do once the model has
been penalized, and the direction of the error is systematic. It manufactures evidence
of misfit out of the estimator rather than the model.

It may be objected that the Hosmer--Lemeshow test is no longer the right diagnostic
anyway, and that a flexible calibration curve with its slope and intercept is to be
preferred. That objection is well taken and it does not dissolve the problem. The
displacement described here is a property of the estimator, not of the grouping: it is
visible in the calibration slope of the glaucoma fit, which is $2.43$ --- an apparent
miscalibration produced entirely by the penalty and not by any defect of the model. We
treat the grouped test because it is the diagnostic that carries a reference
distribution, and therefore the one whose $p$-value can be wrong; a calibration curve
inherits the same displacement but reports it as a picture rather than as a verdict. The
same is true of the smoothing-based test of \citet{lecessie1991}, which is a curve with a
reference distribution attached and would inherit the displacement in both roles.

The paper is organized as follows. Section~\ref{sec:invalid} shows why the reported
$p$-value is not valid and separates two null hypotheses that coincide under maximum
likelihood but not under penalization. Section~\ref{sec:valid} derives the corrected
reference distribution and obtains a valid procedure. Section~\ref{sec:scope} gives the
law governing what such a test can and cannot detect once the linear predictor must be
estimated. Section~\ref{sec:sim} reports the simulation evidence, and
Section~\ref{sec:app} returns to the glaucoma model and states what the corrected
analysis concludes.

\section{Invalidity of the uncorrected test}\label{sec:invalid}

\subsection{Notation}\label{sec:notation}

Let $y_i \in \{0,1\}$ and $\boldsymbol{x}_i \in \mathbb{R}^{p}$, $i = 1,\dots,n$, and
write $X$ for the $n \times (p+1)$ design matrix including the intercept column. The
logistic model specifies
$\pi_i(\bb) = \{1 + \exp(-\boldsymbol{x}_i^{\top}\bb)\}^{-1}$. Penalized estimation
solves
\begin{equation}\label{eq:pen}
  \bbh = \argmax_{\bb}
  \bigl\{ \ell(\bb) - \tfrac{1}{2}\lambda\,\bb^{\top}D\bb \bigr\},
  \qquad D = \mathrm{diag}(0,1,\dots,1),
\end{equation}
the intercept being left unpenalized. The scale of $\lambda$ matters and is a common
source of confusion: the widely used implementation of \citet{friedman2010} minimizes
$-n^{-1}\ell(\bb) + \tfrac12\lambda_{\mathrm{g}}\|\bb_{-0}\|^{2}$, so that
$\lambda = n\lambda_{\mathrm{g}}$. We verified this identity numerically to
$1.6 \times 10^{-9}$ and report both scales throughout.

Write $W = \mathrm{diag}\{\hat\pi_i(1-\hat\pi_i)\}$, $F = X^{\top}WX$, $K = \lambda D$
and $M = F + K$. Let $C$ be the $G \times n$ matrix of indicators of the $G$
equal-frequency groups of the fitted probabilities, $V = CWC^{\top}$, and
\begin{equation}\label{eq:r}
  \boldsymbol{r} = V^{-1/2}C(\boldsymbol{y} - \hat{\boldsymbol{\pi}}), \qquad
  U = V^{-1/2}CWX .
\end{equation}
The Hosmer--Lemeshow statistic is $\|\boldsymbol{r}\|^{2}$ up to the choice of
denominator.

Two bases for the grouped residuals appear throughout, and everything we prove applies to
both. The \emph{decile} statistic is $\|\boldsymbol{r}\|^{2}$ itself, which spends all $G$
directions equally. The \emph{EDGE} statistic of \citet{ebrahim2026edge} instead projects
the residual vector onto the orthogonal polynomials of degree one to three in the
group-mean fitted probability: writing $Z$ for the resulting $G \times 3$ matrix, it is
\begin{equation}\label{eq:edge}
  (Z^{\top}\boldsymbol{r})^{\top}(Z^{\top}Z)^{-1}(Z^{\top}\boldsymbol{r}) .
\end{equation}
It spends three directions rather than $G$, and it is aimed at smooth departures along
the fitted index rather than at arbitrary ones. We report both because they behave
differently in high dimension, and that difference is itself informative.

\subsection{Two null hypotheses that separate under penalization}\label{sec:twonulls}

Two distinct questions are being asked, and ordinary practice does not distinguish them
because ordinarily they cannot be distinguished. Write
\begin{align}
  H_{0}^{\mathrm{str}} &: \ \text{the logistic structure is correct, that is }
    \Pr(y_i = 1 \mid \boldsymbol{x}_i) = \pi_i(\bb_0)
    \text{ for some } \bb_0; \label{eq:h0str}\\
  H_{0}^{\mathrm{cal}} &: \ \text{the deployed predictor is calibrated, that is }
    \mathbb{E}(y_i \mid \hat\pi_i) = \hat\pi_i. \label{eq:h0cal}
\end{align}
Under maximum likelihood at fixed dimension, $\bbh \to \bb_0$, so \eqref{eq:h0str}
implies \eqref{eq:h0cal} asymptotically and the distinction is empty. The qualification
matters and we return to it in Section~\ref{sec:howlarge}(c): once $p/n$ is bounded away
from zero the maximum likelihood estimator is itself displaced, and the two hypotheses
separate even without a penalty. Under penalization it is not:
$\bbh$ converges to a shrunken pseudo-parameter, so a model whose structure is exactly
right delivers a predictor that is exactly \emph{mis}calibrated, by an amount the analyst
chose when selecting $\lambda$.

The second hypothesis is calibration in the sense of \citet{dawid1982}, and sits at the
``moderate'' level of the hierarchy of \citet{vancalster2016}, which separates mean
calibration and calibration in the large from the conditional statement in
\eqref{eq:h0cal}. This matters for reading the output. A test that rejects has told us that
\eqref{eq:h0cal} fails; practitioners read it as evidence that \eqref{eq:h0str} fails.
Under penalization those are different statements, and the whole difficulty of this
paper lies in that gap. Our object is a valid test of \eqref{eq:h0str}. A test of
\eqref{eq:h0cal} is also useful --- it is what one wants before deployment --- but it is
answered by recalibration, not by rebuilding the model.

\subsection{The corrected null distribution}\label{sec:law}

\begin{proposition}\label{prop:law}
Under the regularity conditions given in the supplementary material, with the penalty
contributing $K \succeq 0$ to the curvature,
\[
  \bbh - \bb = M^{-1}\{X^{\top}(\boldsymbol{y}-\boldsymbol{\pi}) - K\bb\} + o_p(n^{-1/2}),
\]
and consequently
$\boldsymbol{r} \adist \mathcal{N}(\boldsymbol{\mu}_K, \Omega_K)$ with
\begin{equation}\label{eq:omega}
  \Omega_K = I_G - UM^{-1}U^{\top} - UM^{-1}KM^{-1}U^{\top}
           = I_G - UM^{-1}(F + 2K)M^{-1}U^{\top},
\end{equation}
\begin{equation}\label{eq:mu}
  \boldsymbol{\mu}_K = UM^{-1}K\bb + \text{(second-order remainder)}.
\end{equation}
\end{proposition}

Two things change at once. They push the test the same way, towards over-rejection, but
they are of very different sizes.

The covariance moves away from its maximum likelihood value $I_G - UF^{-1}U^{\top}$, and
it moves \emph{upward}. Writing $M = F + K$ and using
$F^{-1} - M^{-1} = M^{-1}KF^{-1}$ together with
$M^{-1}(F+2K)M^{-1} = M^{-1} + M^{-1}KM^{-1}$,
\begin{equation}\label{eq:omegaorder}
  F^{-1} - M^{-1}(F+2K)M^{-1} = M^{-1}K\bigl(F^{-1} - M^{-1}\bigr)
  = M^{-1}KF^{-1}KM^{-1} \succeq 0,
\end{equation}
so $\Omega_K \succeq I_G - UF^{-1}U^{\top}$. The classical covariance \emph{understates}
the residual variance once a penalty is applied, which by itself would make a test
referred to it anti-conservative. The effect is real but small: the supplementary
material measures $\tr(\Omega_K)$ exceeding $\tr(I_G - UF^{-1}U^{\top})$ by $0.23$ in
design~A and by $0.92$ in design~B, out of ten groups. A direct decomposition confirms
that it is not what breaks the test: correcting the covariance alone, leaving the
non-centrality in place, moves the rejection rate from $0.994$ to $0.992$ at
$\lambda = 137$ in design~A. It also needs no separate remedy, because
Proposition~\ref{prop:cancel} shows that once the non-centrality is subtracted the
covariance returns to $I_G - UF^{-1}U^{\top}$ exactly.

What is left, and what drives the distortion, is \eqref{eq:mu}. We have
$\mathbb{E}(\boldsymbol{r}) \neq \boldsymbol{0}$ whenever $K\bb \neq \boldsymbol{0}$, so
the correct reference is a \emph{non-central} weighted chi-squared distribution.
Shrinkage displaces the grouped residuals systematically, in a direction fixed by the
penalty, and the classical test reads that displacement as evidence against the model.
The displacement grows with $\lambda$; the classical critical value does not.

\begin{remark}\label{rem:cl}
This is not the Chernoff--Lehmann problem \citep{chernoff1954,moore1975}, in which
estimation shifts the \emph{null} distribution of a chi-squared statistic while leaving
it centred. Here the null distribution acquires a non-zero mean. Moore and Spruill's
unified theory assumes a $\sqrt{n}$-consistent estimator satisfying an orthogonality
condition; a penalized estimator with $\lambda$ growing satisfies neither, which is
precisely why their conclusions do not transfer.
\end{remark}

\subsection{Magnitude of the size distortion}\label{sec:howlarge}

\begin{figure}[!t]
\centering
\includegraphics[width=\textwidth]{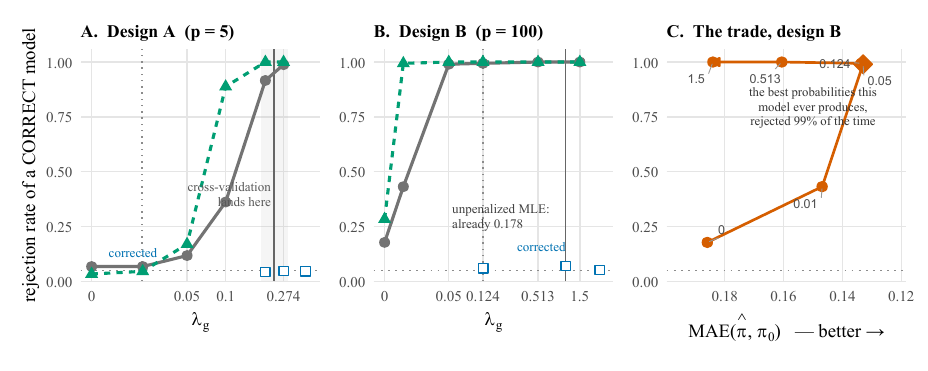}
\caption{Rejection rate of a correctly specified model against the ridge penalty, for design~A
(A) and design~B (B), and against the mean absolute error of the fitted probabilities in
design~B (C). In A and B the solid grey curve is the decile basis and the dashed green
curve the EDGE basis; the dotted and solid vertical rules mark $\hat\lambda_{\min}$ and
$\hat\lambda_{1\mathrm{se}}$, and the shaded band in~A is the interquartile range of
$\hat\lambda_{1\mathrm{se}}$ over replications. Open blue squares are the
shrinkage-corrected test, shown as isolated points because they come from a separate
experiment. In~C the six points are the penalties of design~B in order, the error axis
reversed so that better probabilities lie to the right, and the diamond marks the
smallest error. Based on $500$ replications, Monte Carlo standard error at most $0.022$;
the grid is tabulated in the supplementary material.}
\label{fig:penalty}
\figalttext[Rejection rate of a correct model rises to one along the ridge penalty path,
while the error of the fitted probabilities falls]{Three panels. Panels A and B plot the
rejection rate of a correctly specified model against the ridge penalty on a
pseudo-logarithmic axis, for designs A and B; both the decile and EDGE curves rise from
near the nominal level to 1.00, and vertical rules mark where cross-validation lands,
which in design~A is beyond the point of failure. Isolated open squares near 0.05 show
the corrected test. Panel C plots the same design~B rejection rate against the mean
absolute error of the fitted probabilities, with the error axis reversed so that better
probabilities lie to the right; the trajectory runs up and to the right, and a diamond
marks the penalty at which the probabilities are most accurate and the test rejects
99\% of the time.}
\end{figure}

Four features of Figure~\ref{fig:penalty} deserve comment.

(a) The distortion is not a small-sample curiosity; it reaches certainty. A test whose
size is $1.000$ carries no information at all: its $p$-value is a deterministic function
of the penalty, not of the data.

(b) It occurs at the penalties cross-validation actually chooses. Over $500$
replications the one-standard-error rule exceeded $\lambda_{\mathrm{g}} = 0.05$ in every
replication of both designs, and $\lambda_{\mathrm{g}} = 0.05$ is already the column at
which design~B has reached $0.990$.

(c) In design~B the first column, which is the \emph{unpenalized} fit, is already
invalid, at $0.178$ and $0.284$. This is the phenomenon described by \citet{sur2019} and
\citet{candes2020}: at $p/n$ bounded away from zero the maximum likelihood estimator is
itself biased away from the truth, and \citet{surchen2019} show that in the same regime
the likelihood ratio statistic is no longer asymptotically $\chi^{2}$ but a rescaled one,
so a wrong chi-squared reference in high dimension is not peculiar to grouped tests. It
means no unpenalized comparator is available in that regime, and so a correction cannot be avoided by declining to penalize.

(d) The distortion is worst exactly where the penalty is doing the most good. In
design~B the mean absolute error of the fitted probabilities, $\mathrm{MAE}(\hat{\bpi},
\bpi_0)$, falls from $0.1859$ at $\lambda_{\mathrm{g}} = 0$ to a minimum of $0.1330$ at
$\lambda_{\mathrm{g}} = 0.05$, an improvement of $28\%$; over the same stretch of the
penalty path the Hosmer--Lemeshow rejection rate of a model \emph{known to be correct}
rises from $0.178$ to $0.990$. The penalty that produces the best probabilities this
model will ever produce is the penalty at which the standard check condemns it. We state
this for design~B, where the grid is available; and we note that
$\mathrm{MAE}(\hat{\bpi}, \bpi_0)$ is a simulation quantity, computable only because
$\bpi_0$ is known, which is precisely why the practitioner reaches for a
goodness-of-fit test instead.

\subsection{A one-fit diagnostic screen}\label{sec:screen}

Everything needed to predict the distortion is already computed inside the correction, so
a reader can find out whether their own analysis is affected without running the
procedure at all. Write
\begin{equation}\label{eq:screen}
  N = \|\hat{\boldsymbol{\mu}}\|^{2}, \qquad
  d = \tr\bigl(I_G - UF^{-1}U^{\top}\bigr),
\end{equation}
the squared shrinkage non-centrality and the effective dimension of the grouped residual
vector. The uncorrected statistic is then a non-central weighted chi-squared variable
being referred to a central one, and its size follows by evaluating the two weighted
forms; the eigenvalues of $\Omega_K$ and $I_G - UF^{-1}U^{\top}$ are all that is needed.
This costs one fit and no refits, which is about $150$ times cheaper than the corrected
test itself.

The useful summary is a comparison of two numbers. Across the twelve configurations of
Figure~\ref{fig:penalty}, the uncorrected test rejects a correct model more often than
not once $N$ exceeds $d$, and this rule is right in eleven of the twelve. Averaged over
$200$ designs per cell, the predicted sizes track the measured ones with a mean absolute
error of $0.069$.

Two features of \eqref{eq:screen} explain why dimension is so damaging. $N$ grows like
$\lambda^{2}$ for small $\lambda$ and saturates at $\|U\bb\|^{2}$; and $d$ \emph{falls}
as $p/n$ rises, from $7.97$ in design~A to about $6.0$ in design~B. Dimension moves both
terms the wrong way at once, which is the analytic reason design~B is catastrophic where
design~A is merely bad.

We are explicit about what this screen does not do. It answers ``do I need to worry?'',
not ``is my model wrong?'', and it does not replace the corrected test. Its approximation
is first order, and in design~A at heavy shrinkage it is optimistic --- $0.76$ against a
measured $0.92$. More importantly it cannot see the \citet{sur2019} offset at all: with
no penalty $N = 0$, so it predicts exactly $0.05$ where design~B in fact rejects $17.8\%$
of correct models. The screen therefore \emph{understates} the problem in high dimension
and never overstates it, which is the safe direction for a device whose purpose is to
raise an alarm.

\section{The shrinkage-corrected Hosmer--Lemeshow test}\label{sec:valid}

\subsection{Removal of the shrinkage displacement}\label{sec:remove}

Proposition~\ref{prop:law} suggests an obvious remedy: estimate the shrinkage
non-centrality $\boldsymbol{\mu}_K$ and subtract it. We call the resulting procedure the
\emph{shrinkage-corrected Hosmer--Lemeshow test}, writing \texttt{SC.HL} for it on the
decile grouping and \texttt{SC.EDGE} for it on the EDGE basis of
Section~\ref{sec:notation}. These are one procedure under two choices of basis, not two
tests. The remedy turns out to do more than it was asked to.

\begin{proposition}\label{prop:cancel}
Let $\bbt = \bbh + F^{-1}K\bbh$ and $\hat{\boldsymbol{\mu}} = UM^{-1}K\bbt$. Then, to
first order,
\[
  \var(\boldsymbol{r} - \hat{\boldsymbol{\mu}}) = I_G - UF^{-1}U^{\top},
\]
the ordinary maximum likelihood covariance.
\end{proposition}

We give this proof in the text, rather than in the supplement with the others, because it
is short and because it is the reason the method needs no new reference distribution.

\begin{proof}
Since $I + F^{-1}K = F^{-1}(F+K) = F^{-1}M$, the debiased estimate is
$\bbt = F^{-1}M\bbh$. Substituting the expansion of Proposition~\ref{prop:law} and using
$F^{-1}M\bb = \bb + F^{-1}K\bb$, the two penalty terms cancel and
\begin{equation}\label{eq:btilde}
  \bbt = \bb + F^{-1}X^{\top}(\by-\bpi) + o_p(n^{-1/2}),
\end{equation}
which is the first-order expansion of the \emph{unpenalized} maximum likelihood
estimator in the regime of the regularity conditions. Where that estimate does not exist,
as on the data of Section~\ref{sec:intro}, $\bbt$ should be read not as an estimator in
its own right but as a one-step undoing of the shrinkage the penalty imposed. The same expansion writes the grouped
residuals as $\boldsymbol{r} = A(\by-\bpi) + UM^{-1}K\bb + o_p(1)$ with
$A = V^{-1/2}C - UM^{-1}X^{\top}$. Subtracting $\hat{\boldsymbol{\mu}} = UM^{-1}K\bbt$
removes the non-random term exactly and leaves
\[
  \boldsymbol{r} - \hat{\boldsymbol{\mu}}
  = \bigl\{V^{-1/2}C - UM^{-1}(I + KF^{-1})X^{\top}\bigr\}(\by-\bpi) + o_p(1).
\]
The cancellation is now algebraic:
\begin{equation}\label{eq:cancel}
  M^{-1}(I + KF^{-1}) = M^{-1}(F+K)F^{-1} = M^{-1}MF^{-1} = F^{-1}.
\end{equation}
Writing $A^{*} = V^{-1/2}C - UF^{-1}X^{\top}$ and using $\var(\by-\bpi) = W$ together
with $V^{-1/2}CWC^{\top}V^{-1/2} = I_G$, $V^{-1/2}CWX = U$ and $X^{\top}WX = F$ gives
$A^{*}WA^{*\top} = I_G - UF^{-1}U^{\top}$.
\end{proof}

Two things are worth drawing out. The estimation noise contributed by the plug-in exactly
cancels the penalization of the covariance, and \eqref{eq:cancel} shows this is an
identity rather than an approximation. So the practitioner does not have to learn a new
reference distribution: once the displacement is removed, the familiar one is correct
again. And \eqref{eq:btilde} says what $\bbt$ is --- the one-step correction that undoes,
to first order, exactly the shrinkage the penalty imposed. That is also the one place the
procedure is fragile, since it inverts $F$.

\begin{remark}\label{rem:smooth}
Nothing in either proposition uses the linearity of the ridge penalty. The proof needs
the penalty only through its contribution to the estimating equation, so if a smooth
penalty contributes a known deterministic $\boldsymbol{a}(\bb)$ with
$K(\bb) = -\partial \boldsymbol{a}/\partial\bb$, then
$\boldsymbol{\mu} = UM^{-1}\boldsymbol{a}(\bb)$ and
$\bbt = \bbh + F^{-1}\boldsymbol{a}(\bbh)$, and the cancellation \eqref{eq:cancel} goes
through unchanged, because it only ever used $M^{-1}(F+K)F^{-1} = F^{-1}$. Ridge is the
case $\boldsymbol{a}(\bb) = K\bb$. This covers the Liu estimator, the ridge part of the
elastic net, and any Bayesian maximum a posteriori fit with a smooth prior. It also
covers, formally, Firth's correction \citep{firth1993}, whose score contribution is the
Jeffreys term --- which matters here because Firth, not ridge, is the standard response
to separation in clinical prediction \citep{puhr2017}. We state that as a formal
consequence and not as a verified one: Firth's $\boldsymbol{a}(\bb)$ is not linear and
its $K(\bb)$ is not constant, so the same first-order remainder caveat applies with more
force, and we have not run the numerical check that would let us recommend it.
\end{remark}

There is also a reading of $\bbt$ that connects this construction to a literature it
might otherwise seem separate from. The penalized stationarity condition is
$X^{\top}(\by - \hat\bpi) = K\bbh$. Substituting that into the generic one-step debiased
estimator $\bbh + F^{-1}X^{\top}(\by - \hat\bpi)$ returns $\bbh + F^{-1}K\bbh$, which is
$\bbt$. So $\bbt$ is the debiased estimator of \citet{zhangzhang2014} and
\citet{vandegeer2014}, specialized to ridge --- and in the ridge case two things that are
approximations in that literature become exact: the score residual is a deterministic
function of $\bbh$, and $F^{-1}$ is available directly rather than through a
nodewise-lasso approximation. That is why the cancellation of
Proposition~\ref{prop:cancel} is an identity here and not an approximation.

It is worth saying how that fragility does and does not show up, because the obvious
diagnostic is the wrong one. On the glaucoma data the cross-product matrix has condition
number about $9\times10^{7}$, and $\|\bbt\|$ exceeds $\|\bbh\|$ by a factor between $618$
and $1637$ across the four penalties. Read as an estimate, $\bbt$ is meaningless there.
But $\bbt$ is never read as an estimate: it enters only through the generator
$\boldsymbol{\pi}(\bbt)$, and the generator is well behaved. At the cross-validated
penalty its fitted probabilities span $[0.0009, 0.997]$ with standard deviation $0.34$,
not one of the $196$ observations lies within $10^{-6}$ of $0$ or $1$, and it implies
$96.7$ events against the $98$ observed. The large components of $\bbt$ lie in directions
in which the design has almost no variance, so they inflate its norm and leave
$X\bbt$ almost unchanged. The quantity a practitioner should check before trusting this
procedure is therefore the spread of the generator, which is computable on any dataset,
and not the norm of $\bbt$.

\subsection{Prepivoting}\label{sec:prepivot}

The first-order correction is not enough at the penalties cross-validation selects. It
removes the anti-conservative failure but leaves the test conservative, because
$\hat{\boldsymbol{\mu}}$ is itself estimated at a penalty large enough that the
first-order account of its error is inadequate.

We therefore refer the corrected statistic to a reference built from itself. The
procedure is:

(a) fit the penalized model at $\lambda$ and compute
    $S = \|\boldsymbol{r} - \hat{\boldsymbol{\mu}}\|^{2}$;

(b) form the generator $\boldsymbol{\pi}(\bbt)$ from the debiased estimate;

(c) for $b = 1,\dots,B$, draw $\boldsymbol{y}^{*}_{b}$ from that generator, refit the
    penalized model \emph{at the same $\lambda$}, and recompute the \emph{same corrected}
    statistic $S^{*}_{b}$;

(d) report $\{1 + \#(S^{*}_{b} \geq S)\}/(B+1)$.

Because the statistic is corrected in both the real and the bootstrap world, the
reference inherits the same finite-sample imperfections as the observed value and
cancels them. This is the prepivoting principle of \citet{beran1987,beran1988}: one does
not need the statistic to be asymptotically pivotal, only to be referred to its own
sampling distribution under a generator that is correct to the order that matters.

Three parts of this procedure carry the argument, and two of them have a near neighbour
that fails. The generator must be $\boldsymbol{\pi}(\bbt)$: generating from
$\boldsymbol{\pi}(\bbh)$ builds the shrinkage into the null world, and the test becomes
conservative. The statistic must be corrected in both worlds, since an uncorrected
bootstrap statistic and a corrected observed one are not the same quantity. And we hold
$\lambda$ fixed in the bootstrap world, so that the null world's penalty is not itself
random. The catalogue of neighbouring constructions, with the size of each failure, is in
the supplementary material. It is also worth stating that an oracle which knows $\bb_0$
cannot be built on a bootstrap reference at all, because the truth of the bootstrap
world is $\bbt$ and not $\bb_0$; oracle comparisons in
Section~\ref{sec:scope} therefore take their critical value by direct Monte Carlo from
the true model instead.

\section{Detectable and undetectable departures}\label{sec:scope}

A valid test is not automatically a useful one. This section gives the law that governs
its power, and we present it as a contribution rather than as a limitation, because it
applies to every test in this family and not only to ours.

\subsection{The visible effect size}\label{sec:tau}

Grouping is performed on the \emph{fitted} linear predictor $\hat\eta$, so the only part
of a departure that a grouped test can see is the part that survives conditioning on
$\hat\eta$. For a departure $\gamma\,g(\boldsymbol{x})$ added to the linear predictor,
define the visible effect size
\begin{equation}\label{eq:tau}
  \tau = \gamma \, \frac{\sd[\mathbb{E}\{g(\boldsymbol{x}) \mid \hat\eta\}]}
                         {\sd(\hat\eta)} .
\end{equation}
The nominal coefficient $\gamma$ is not comparable across designs --- the same $\gamma$
means different things when $p = 5$ and when $p = 100$ --- whereas $\tau$ is. When power
is plotted against $\nu\sqrt{n}$, where $\nu = \tau\,\sd(\hat\eta)$, curves from designs
of very different dimension collapse onto one another
(Figure~\ref{fig:power}, panel~A; mean absolute gap $0.017$ against a Monte Carlo
standard error of $0.023$ for a single difference).

There is a subtlety worth making explicit, because it corrects an expectation one
naturally forms from Proposition~\ref{prop:atten}. One might expect a residual
dimensional offset to remain visible on a $\tau$ axis. It does not, and should not: the
numerator of \eqref{eq:tau} already contains the factor $\rho^{k}$, so $\tau$ absorbs
the attenuation by construction. The law of Section~\ref{sec:atten} is untouched; what
it governs is how much $\tau$ a given $\gamma$ buys, not the shape of the curve once
$\tau$ is on the axis.

There is also a limit beyond which \emph{any} grouped test may lose power as the
departure grows, and it is worth stating because we have not seen it stated. A cubic
departure in the index standardizes to
$Z + \mathrm{rel}\cdot(Z^{3}-3Z)/\sqrt{6}$ with $\mathrm{rel} = \gamma/\sd(\eta_0)$, whose
derivative at $Z = 0$ is $1 - 3\,\mathrm{rel}/\sqrt{6}$. Once
$\mathrm{rel} \geq \sqrt{6}/3 = 0.82$ the true event probability is no longer monotone in
the linear predictor, so sorting patients by predicted risk no longer sorts them by the
departure, and the power of any test that groups on the index can fall as the departure
grows. Of the cells in Figure~\ref{fig:power}, one point of the design~A cubic series
lies beyond this threshold ($\mathrm{rel} = 0.97$) and three of the design~B cubic series
do ($1.05$, $1.57$, $2.19$). Power nevertheless rises monotonically across the plotted
range, so the threshold bounds the interpretation rather than invalidating the curves.
This is a property of grouping on a monotone summary, shared by the Hosmer--Lemeshow
family and by every grouped test; we give its threshold because it is easy to cross
without noticing.

\subsection{The attenuation law}\label{sec:atten}

\begin{proposition}\label{prop:atten}
Write $\eta^{\ast}$ and $\hat\eta^{\ast}$ for the true and the fitted linear predictor,
each standardized to mean zero and unit variance, and suppose the pair is jointly
Gaussian with correlation $\rho$. Since
$\mathbb{E}\{\mathrm{He}_k(\eta^{\ast}) \mid \hat\eta^{\ast}\} =
\rho^{k}\mathrm{He}_k(\hat\eta^{\ast})$ for the Hermite polynomials $\mathrm{He}_k$, a
departure equal to the degree-$k$ component of the true linear predictor survives
grouping by the fitted predictor with amplitude $\rho^{k}$, and its non-centrality is
attenuated by $\rho^{2k}$.
\end{proposition}

The consequence is sharp and, once stated, unsurprising. Estimation error in the index
costs a linear departure a factor $\rho^{2}$, a quadratic departure $\rho^{4}$, and a
cubic departure $\rho^{6}$. At $\rho = 0.9$ --- a fitted index that most analysts would
regard as excellent --- a cubic departure retains only $53\%$ of its non-centrality; at
$\rho = 0.7$, only $12\%$. Higher-order departures are therefore intrinsically hard to
see through an estimated index, for \emph{any} test that groups on that index.

\subsection{Dilution and index-estimation noise}\label{sec:mechanisms}

It is tempting to collapse the high-dimensional difficulty into a single statement about
$p/n$. That would be wrong, and the distinction is practically useful because the two
mechanisms have different remedies.

(a) \emph{Dilution} is a property of the design, not of the dimension. A departure
attached to one coordinate, such as $\gamma(x_1^{2}-1)$, contributes to the index in
proportion to that coordinate's share of the index variance. In our design~B that share
is small, and the visible \emph{variance} of the departure is accordingly smaller by a
factor of $16$ than in design~A. Dilution would occur at the same magnitude in a
low-dimensional design in which one coordinate happened to carry a small weight.

(b) \emph{Index-estimation noise} is the dimensional mechanism, and it is
Proposition~\ref{prop:atten}: as $p/n$ grows, $\rho$ falls, and every departure is
attenuated by $\rho^{2k}$.

The two are separable by experiment, and the experiment separates them cleanly. Group
the residuals by the \emph{true} index instead of the fitted one, holding everything else
fixed, and calibrate each arm by direct Monte Carlo from the true model so that both are
exactly correct in size by construction. Whatever power then changes is attributable to
the grouping alone. In design~B this gives Table~\ref{tab:oracle}.

\begin{table}[!t]
\centering
\caption{Grouping on the fitted index against grouping on the true one, design~B, with
each arm calibrated to exact size by direct Monte Carlo so that the comparison is of
power only. Knowing the index transforms the index-aligned departure and does nothing at
all for the coordinate-wise one, which is the two mechanisms separated.}
\label{tab:oracle}
\begin{tabular}{llrrrr}
\toprule
& & \multicolumn{2}{c}{Fitted index} & \multicolumn{2}{c}{True index} \\
\cmidrule(lr){3-4}\cmidrule(lr){5-6}
Departure & $\gamma$ & decile & EDGE & decile & EDGE \\
\midrule
cubic (index-aligned)     & 1.5 & 0.092 & 0.046 & 0.623 & 0.633 \\
quadratic (coordinate-wise) & 1.0 & 0.059 & 0.061 & 0.067 & 0.059 \\
\bottomrule
\end{tabular}
\end{table}

The contrast is the whole argument of this section. For the index-aligned departure,
knowing the true index raises power from $0.092$ to $0.623$ on the decile basis and from
$0.046$ to $0.633$ on the EDGE basis, a factor of nearly fourteen. For the
coordinate-wise departure it changes nothing: $0.059$ to $0.067$, and on the EDGE basis
$0.061$ to $0.059$. So mechanism~(b) is real and large, but it is not what limits the
coordinate-wise case; there the departure is diluted before any index is estimated, and
an oracle who knew the index exactly would be no better off. Conversely an index-aligned
departure suffers no dilution at all, and its power curves collapse across designs once
expressed in $\nu\sqrt{n}$.

The honest scope statement is therefore not that the test lacks power in high
dimensions. It is that a grouped test on an estimated index detects index-aligned
departures at a rate governed by $\nu\sqrt{n}$ in any dimension, and detects
coordinate-wise departures at a rate that also depends on how much of the index that
coordinate carries.

\subsection{Comparison with residual-prediction tests}\label{sec:rivals}

Tests for high-dimensional generalized linear models built by other routes
\citep{guo2016,shah2018,jankova2020,javanmard2024} are usually described as alternatives
to a test like ours. We ran the comparison rather than assert it. In design~B at
$\lambda = 416$ with $G = 10$, fully paired so that every method sees the same dataset in
every replicate, the residual-prediction test of \citet{guo2016} beats ours decisively on
the coordinate-wise quadratic departure --- $0.710$ at $\gamma = 1.0$ and $1.000$ at
$\gamma = 1.5$, against our $0.055$ to $0.090$ --- and the two are level on the cubic.

That result is worth reporting for its own sake, but it is not the interesting one. The
two tests are not noisy versions of each other: \emph{they reject on different datasets}.
Across all six configurations the number of replicates in which both tests reject is
between $0$ and $15$ out of $200$, and the paired differences are $-0.655$ and $-0.925$
with $p < 2\times10^{-16}$. A residual-prediction method asks whether the model is right
as a function of $\boldsymbol{x}$; ours asks whether the deployed predictor is calibrated
along its own index. Those are different questions, and the data say so.

Our loss on the coordinate-wise departure is therefore confirmation of this paper's own
dilution mechanism, arriving from outside it: $x_1$ carries $3.4\%$ of the variance of
the index at $p = 100$, so a departure attached to $x_1$ is nearly invisible to any test
that groups on the index, and plainly visible to one that does not. An analyst who
suspects a specific covariate should use a residual-prediction test. One who wants to
know whether the risks the model reports can be read as risks should use this one.

Three points of software provenance, because they are checkable. The Shah--B\"uhlmann
package is \texttt{RPtests}, and it targets the Gaussian linear model; on a binary
response it is off-label, and we ran it only so that its size could speak. The
\citet{guo2016} implementation \texttt{GRPtests} was removed from the Comprehensive R
Archive Network and archived on 8 May 2022; we installed it from the archived tarball,
and we ran it both at its package default and at a single split, so that it is not
handicapped by a tuning choice. GRASP has no package and was not attempted.

\section{Simulation study}\label{sec:sim}

Two designs are used throughout. Design~A has $n = 500$, $p = 5$ and independent
standard normal covariates; design~B has $n = 400$, $p = 100$ and an autoregressive
covariate correlation of $0.7$, giving $p/n = 0.25$. In both, the null model is logistic
with the stated coefficients, so every rejection under the null is a false one. Full
specifications, seeds and code are in the supplementary material.

\subsection{Size}\label{sec:size}

\begin{table}[!t]
\centering
\caption{Rejection rate of the \emph{prepivoted corrected} test under a correctly
specified model, at nominal level $0.05$, with $149$ bootstrap replicates. Design~A uses
$2000$ Monte Carlo replications and design~B uses $1000$, except the $\lambda = 50$ row,
which is pooled over three independent replication blocks totalling $4000$ and is marked
$\dagger$. Penalties are on the theory scale $\lambda$; the corresponding
$\lambda_{\mathrm{g}}$ is $\lambda/n$. The uncorrected test on these same six cells
rejects between $0.92$ and $1.00$. With $B = 149$ the bootstrap $p$-value takes the
values $k/150$, so $0.05$ is not itself attainable; the adjacent attainable levels are
$7/150 = 0.0467$ and $8/150 = 0.0533$.}
\label{tab:size}
\begin{tabular}{llrrrr}
\toprule
& & \multicolumn{2}{c}{\texttt{SC.HL}} & \multicolumn{2}{c}{\texttt{SC.EDGE}} \\
\cmidrule(lr){3-4}\cmidrule(lr){5-6}
Design & $\lambda$ & rate & s.e. & rate & s.e. \\
\midrule
A & 100  & 0.0435 & 0.0046 & 0.0500 & 0.0049 \\
A & 137  & 0.0475 & 0.0048 & 0.0505 & 0.0049 \\
A & 200  & 0.0465 & 0.0047 & 0.0505 & 0.0049 \\
\midrule
B & 50$^{\dagger}$ & 0.0595 & 0.0037 & 0.0370 & 0.0030 \\
B & 416  & 0.0710 & 0.0081 & 0.0540 & 0.0071 \\
B & 1000 & 0.0520 & 0.0070 & 0.0610 & 0.0076 \\
\bottomrule
\end{tabular}
\end{table}

Table~\ref{tab:size} is the central empirical claim of the paper. Against an uncorrected
test that rejects correct models between $92\%$ and $100\%$ of the time on these same
six configurations, the prepivoted corrected test holds its nominal level. In design~A
it does so at every penalty on both bases, with standard errors of $0.005$.

One cell requires comment rather than concealment. On the EDGE basis in design~B at the
lightest penalty, $\lambda = 50$, the test is detectably conservative: $0.0370$ with an
exact $95\%$ interval of $[0.0314, 0.0433]$, which is $4.4$ standard errors below
nominal. We report this cell at four times the replication of the others because our
first estimate of it, from a single block of $1000$, was $0.0270$, which would have
placed it outside the range $[0.03, 0.08]$ we fixed in advance as acceptable. The pooled
estimate lies inside that range; it is the superseded single block that did not. Two further independent blocks did not
reproduce that value, the three are homogeneous ($p = 0.09$), and the pooled estimate is
the one given. The episode is a caution about reading a single block at a probability
this small, and we apply the same standard to any other cell that falls outside the
range.

The conservatism that survives is genuine but modest, and it is not uniform --- the same
basis returns $0.0540$ at $\lambda = 416$ and $0.0610$ at $\lambda = 1000$ --- so it is
confined to the light-penalty end of the high-dimensional design and should not be
described as a general property of the basis. It does mean that where the EDGE basis is
compared with the decile basis in design~B, a size-adjusted comparison is the fair one,
and we report both.

\subsection{Power}\label{sec:power}

\begin{table}[!t]
\centering
\caption{Power of the prepivoted corrected test on the decile and EDGE bases, paired on
identical fits and refits with the classical test on the maximum likelihood fit. Two
departures are used: an omitted quadratic $\gamma(x_1^{2}-1)$ and an omitted interaction
$\gamma x_1 x_2$. Based on $2000$ replications (Monte Carlo standard error at most
$0.0112$). The $\gamma = 0$ row is size, and is the same simulation for both departures.
No valid maximum likelihood comparator is available in design~B. Size-adjusted power is given in the
supplementary material; it differs from the entries below by at most $0.02$ and reverses
no comparison. Design~A is at $\lambda = 100$ and design~B at $\lambda = 416$. The
$\gamma = 0$ rows come from a different run than the corrected-size table
(\texttt{T25\_power\_highB.R}, stream $20260810$, against \texttt{T23\_rerun.R}, stream
$20260807$); the two agree to less than one Monte Carlo standard error.}
\label{tab:power}
\begin{tabular}{lllrrrr}
\toprule
& & & \multicolumn{2}{c}{Corrected} & \multicolumn{2}{c}{Classical, on the MLE} \\
\cmidrule(lr){4-5}\cmidrule(lr){6-7}
Design & Departure & $\gamma$ & \texttt{SC.HL} & \texttt{SC.EDGE} & decile & EDGE \\
\midrule
A & ---          & 0.00 & 0.0445 & 0.0460 & 0.0440 & 0.0540 \\
\addlinespace
A & quadratic    & 0.25 & 0.0760 & 0.1150 & 0.0890 & 0.1205 \\
A & quadratic    & 0.50 & 0.1475 & 0.2615 & 0.1715 & 0.2675 \\
A & quadratic    & 0.75 & 0.2015 & 0.3755 & 0.2350 & 0.3835 \\
A & quadratic    & 1.00 & 0.2700 & 0.4150 & 0.2960 & 0.4080 \\
\addlinespace
A & interaction  & 0.25 & 0.0595 & 0.0750 & 0.0630 & 0.0870 \\
A & interaction  & 0.50 & 0.1055 & 0.1885 & 0.1290 & 0.2065 \\
A & interaction  & 0.75 & 0.1940 & 0.3555 & 0.2275 & 0.3780 \\
A & interaction  & 1.00 & 0.2835 & 0.4615 & 0.3265 & 0.4885 \\
\addlinespace
B & ---          & 0.00 & 0.0640 & 0.0495 & --- & --- \\
B & quadratic    & 0.50 & 0.0595 & 0.0485 & --- & --- \\
B & quadratic    & 1.00 & 0.0635 & 0.0600 & --- & --- \\
B & quadratic    & 1.50 & 0.0740 & 0.0630 & --- & --- \\
\bottomrule
\end{tabular}
\end{table}

The comparison that matters is not corrected-versus-uncorrected under the null, where
the uncorrected test is disqualified, but corrected-versus-classical where the classical
test is legitimate. We therefore compare the corrected test on a penalized fit with the
classical test on a maximum likelihood fit in design~A, where the latter exists. Paired
over the eight design~A alternatives, the correction costs what
Figure~\ref{fig:costgain} and Table~\ref{tab:cost} report, and it is \emph{not} free. On
the decile basis it costs about two and a half
percentage points, or thirteen per cent of the classical test's power, and at $2000$
replications that is unambiguous. We record this plainly because an earlier version of
this experiment, at $200$ replications, returned $-0.010$ with $p = 0.10$ and would have
supported the more comfortable claim that the cost is indistinguishable from zero. It was
simply too small to see the effect.

\begin{table}[!t]
\centering
\caption{What the correction costs in power. Each entry is the mean difference between
the corrected test on a penalized fit and the classical test on a maximum likelihood fit,
paired over the eight design~A alternatives of Table~\ref{tab:power} and computed from
$2000$ replications per cell. The cost is real on both bases and larger on the decile
basis.}
\label{tab:cost}
\begin{tabular}{lrrr}
\toprule
Test & mean difference & $t$ ($7$ d.f.) & relative \\
\midrule
\texttt{SC.HL}   & $-0.0250$ \ (s.e.\ $0.0044$) & $-5.70$, $p = 0.0007$ & $-13.0\%$ \\
\texttt{SC.EDGE} & $-0.0115$ \ (s.e.\ $0.0038$) & $-3.00$, $p = 0.020$ & $-3.9\%$ \\
\bottomrule
\end{tabular}
\end{table}

Two things put the cost in proportion. It is a cost paid only where there is a choice. In design~B no valid maximum likelihood
comparator is available, because the unpenalized fit is itself invalid there; and on the
glaucoma data the maximum likelihood estimate does not exist at all. In both the
classical test is unavailable at any price, and the comparison has no second column. And it is much smaller than the gain available from the
basis. Switching from the decile to the EDGE basis under the same correction raises mean
power across these eight alternatives from $0.167$ to $0.281$, an increase of $0.114$ or $68\%$ --- about four and
a half times what the correction costs. An analyst who wants the power back does
not need a different test; they need a better basis.

\begin{figure}[!t]
\centering
\includegraphics[width=\textwidth]{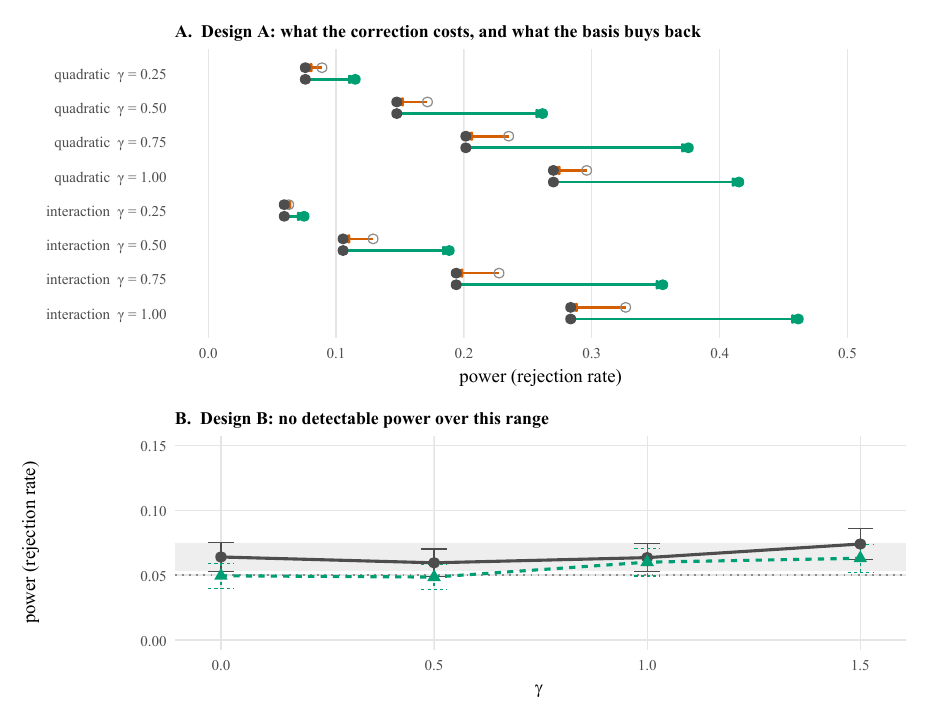}
\caption{Power of the classical test on a maximum likelihood fit, of \texttt{SC.HL}, and
of \texttt{SC.EDGE}. In panel~A each alternative
occupies two half-rows: above, an orange arrow from the classical test (open circle) to
\texttt{SC.HL} (filled grey circle); below, a green arrow from that point to
\texttt{SC.EDGE} (filled green circle). The two arrows are drawn with identical weight.
Panel~B is design~B on a tight vertical scale, with two-standard-error bars and a grey
band at the $\gamma = 0$ rejection rate plus or minus two standard errors. Based on
$2000$ replications per cell.}
\label{fig:costgain}
\figalttext[Short orange cost arrows against much longer green gain arrows]{Panel A is a
horizontal dot plot of eight alternatives. For each, a short orange arrow points left
from the classical test to SC.HL, and a much longer green arrow points right
from SC.HL to SC.EDGE. Panel B plots design B power
against gamma on a scale from zero to 0.15; all points lie within a shaded band marking
the null rejection rate.}
\end{figure}

Design~B tells a blunter story. Over the range $\gamma \leq 1.5$ the corrected test has no
detectable power at all: the decile rejection rate moves from $0.0640$ to $0.0740$
between $\gamma = 0$ and $\gamma = 1.5$, a change of $0.010$ against a standard error of
$0.0075$. This too corrects the smaller study, which showed $0.040$ rising to $0.095$ and
looked like the beginning of a power curve. It was noise. The departure has to be made several times larger before design~B responds at all, for
the dilution reason given in Section~\ref{sec:mechanisms}.

Figure~\ref{fig:power} shows what governs power once dimension enters. It is built on a
different calibration from Tables~\ref{tab:size}--\ref{tab:cost}, and we say so because
the difference matters for reading it. Its arms are calibrated by direct Monte Carlo from the true model rather than by
prepivoting --- $2000$ null draws per cell, with the critical value taken at the $95$th
percentile --- so every curve starts from exactly $0.05$ by construction. That is the right design for the question the figure asks --- which
departures a grouped test can see, and on what scale --- because it removes the
calibration of the test as a competing explanation for any difference between curves. It
does mean the figure is not a display of the prepivoted procedure itself; for that, see
Tables~\ref{tab:size} and~\ref{tab:power}.

\begin{figure}[!t]
\centering
\includegraphics[width=\textwidth]{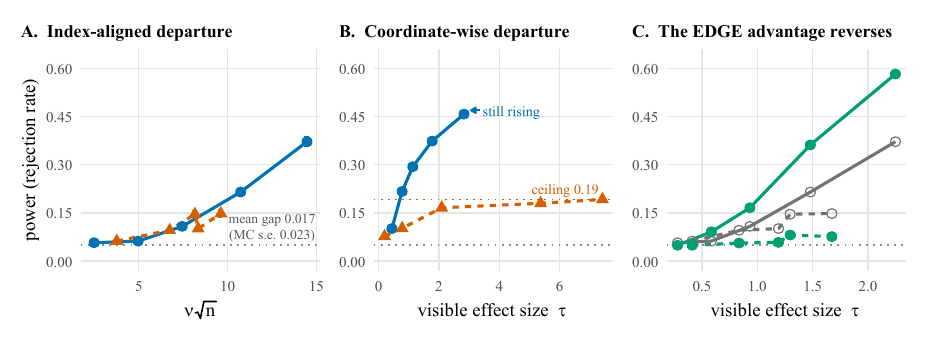}
\caption{Power of the corrected test against $\nu\sqrt{n}$ for an index-aligned departure (A) and
against the visible effect size $\tau$ for a coordinate-wise departure (B) and for the
two bases (C). In A and B the solid blue curve with circles is design~A and the dashed
orange curve with triangles design~B; in C the green curve with filled circles is the
EDGE basis and the grey curve with open circles the decile basis, solid for design~A and
dashed for design~B. Rejection rates use a critical value obtained by direct Monte Carlo
from the true model rather than by prepivoting, so every cell is exactly $5\%$ under the
null by construction; $B_1 = 1000$ replications per cell, Monte Carlo standard error at
most $0.016$, shown as error bars in panel~A. In panel~A the third and fourth design~B
points sit at almost the same value of $\nu\sqrt{n}$ ($8.16$ and $8.34$) because $\nu$ is
not monotone in $\gamma$.}
\label{fig:power}
\figalttext[Three panels showing that power collapses onto one curve for index-aligned
departures but not for coordinate-wise ones]{Three line charts. Panel A plots power
against nu times the square root of n for an index-aligned cubic departure; the design~A
and design~B curves lie on top of one another, with error bars overlapping throughout.
Panel B plots power against the visible effect size tau for a coordinate-wise quadratic
departure; the design~A curve continues to rise while the design~B curve flattens at
about 0.19. Panel C plots power against tau for both the EDGE and decile bases; EDGE lies
above decile over most of design~A and below it throughout design~B.}
\end{figure}

\section{Reanalysis of the glaucoma model}\label{sec:app}

We now return to the model of Section~\ref{sec:intro} and apply the corrected test.

\subsection{The source of the uncorrected rejection}\label{sec:whatitsaw}

Figure~\ref{fig:calib} shows the calibration of the ridge fit at the cross-validated
penalty. The pattern is not random. In the lowest deciles of predicted risk the observed
event rate lies \emph{below} the diagonal --- risk is predicted too high --- and in the
highest deciles it lies \emph{above} it, where risk is predicted too low. The fitted
probabilities are therefore pulled systematically towards the centre, in nine of the ten
deciles and in the direction the penalty dictates; the sixth sits on the diagonal, at the
crossing point. The apparent calibration slope is $2.43$,
meaning the linear predictor would have to be multiplied by roughly two and a half
before the fitted probabilities agreed with the observed rates.

That compression is not a defect of the model. It is what a penalty does, and it is the
price paid for obtaining any estimate at all in a design where the maximum likelihood
estimate does not exist. It is the displacement that Proposition~\ref{prop:law}
describes. The uncorrected test sees it, has no way to distinguish it from
misspecification, and reports the model as wrong. We say it is the displacement
Proposition~\ref{prop:law} describes rather than claim that the proposition applies here:
its conditions hold the dimension fixed, and these data have $p/n = 0.32$. What the
proposition supplies on this dataset is the mechanism and its direction; what supplies
the $p$-values is the prepivoted procedure, whose validity in this regime rests on the
simulation evidence of Section~\ref{sec:sim}.

\begin{figure}[!t]
\centering
\includegraphics[width=0.72\textwidth]{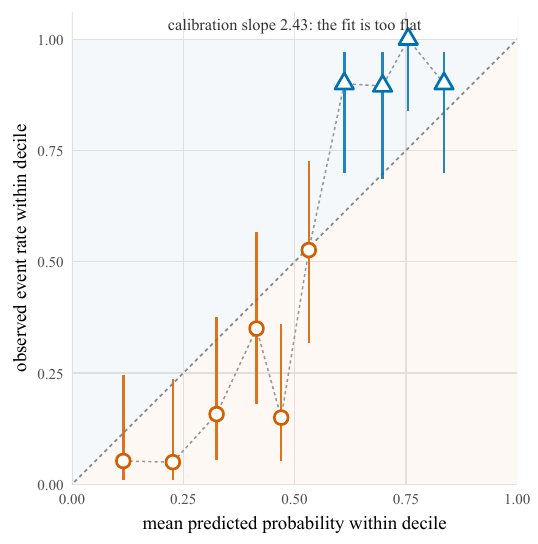}
\caption{Observed event rate against mean predicted probability within deciles of the ridge fit on
the glaucoma data, at the cross-validated penalty. Deciles below the diagonal, drawn as
orange circles, have their risk predicted too high; those above it, drawn as blue
triangles, have it predicted too low. Vertical bars are Wilson intervals and the dashed
line is equality.}
\label{fig:calib}
\figalttext[Calibration plot showing predicted risks compressed towards the centre]{
Calibration plot of observed event rate against mean predicted probability for ten
deciles of the ridge fit, with Wilson intervals and a diagonal reference line. The six
lowest deciles lie below the diagonal, meaning risk is predicted too high, and the four
highest lie above it, meaning risk is predicted too low; the sixth lies almost exactly on
it, observed 0.526 against predicted 0.532. The apparent calibration slope is 2.43.}
\end{figure}

\subsection{The corrected analysis}\label{sec:corrected}

\begin{table}[!t]
\centering
\caption{Uncorrected and corrected $p$-values for the glaucoma model, over four
penalties and three group counts. The uncorrected test rejects in all twelve
configurations; the corrected test does not. Penalties are on the theory scale, with
$\hat\lambda_{1\mathrm{se}} = 295.3$ and $\hat\lambda_{\min} = 19.9$ selected by
ten-fold cross-validation. The uncorrected column refers the grouped statistic to the
maximum likelihood covariance $I_G - UF^{-1}U^{\top}$ by Monte Carlo, which is the
strongest form of the incumbent; entries below $10^{-5}$ are the resolution limit of
$2\times10^{5}$ draws. Corrected $p$-values use $499$ bootstrap replicates.}
\label{tab:glaucoma}
\small
\setlength{\tabcolsep}{4pt}
\begin{tabular}{lrrrrrrrrr}
\toprule
& \multicolumn{3}{c}{Uncorrected}
& \multicolumn{3}{c}{\texttt{SC.HL}}
& \multicolumn{3}{c}{\texttt{SC.EDGE}} \\
\cmidrule(lr){2-4}\cmidrule(lr){5-7}\cmidrule(lr){8-10}
$\lambda$ & $G=5$ & $G=10$ & $G=20$ & $G=5$ & $G=10$ & $G=20$
          & $G=5$ & $G=10$ & $G=20$ \\
\midrule
$\hat\lambda_{\min}$ & 0.0003 & 0.0005 & 0.0053
  & 0.014 & 0.048 & 0.140 & 0.456 & 0.096 & 0.114 \\
$\tfrac12\hat\lambda_{1\mathrm{se}}$ & $<10^{-5}$ & $<10^{-5}$ & 0.0001
  & 0.032 & 0.070 & 0.062 & 0.094 & 0.032 & 0.106 \\
$\hat\lambda_{1\mathrm{se}}$ & $<10^{-5}$ & $<10^{-5}$ & $<10^{-5}$
  & 0.036 & 0.026 & 0.172 & 0.218 & 0.034 & 0.078 \\
$2\hat\lambda_{1\mathrm{se}}$ & $<10^{-5}$ & $<10^{-5}$ & $<10^{-5}$
  & 0.294 & 0.276 & 0.678 & 0.500 & 0.038 & 0.096 \\
\bottomrule
\end{tabular}
\end{table}

Table~\ref{tab:glaucoma} reports both tests across four penalties and three choices of
the number of groups. Three things should be said about it, in order of how much they
matter.

(a) The uncorrected test is uninformative here. It rejects in every one of the twelve
configurations, with $p$ below $0.006$ throughout and below $10^{-5}$ at and beyond the
cross-validated penalty. Given the size distortion documented in Section~\ref{sec:sim}
--- in simulation designs of comparable $p/n$ at comparable shrinkage, rejection of
correctly specified models between $92\%$ and $100\%$ of the time
(Table~\ref{tab:size}) --- a small $p$-value from this test carries almost no information
about this model.

We use the sharpest available form of the uncorrected test, referring the grouped
statistic to the maximum likelihood covariance rather than to $\chi^2_{G-2}$. The
distinction matters more than one would expect, and it illustrates the paper's own point
from a different angle. The classical degrees-of-freedom count would be $G-2 = 8$ at ten
groups, so it allows two decile directions to be absorbed by estimation. The trace of
$I_G - UF^{-1}U^{\top}$ on these data is $4.84$ at the cross-validated penalty, and lies
between $4.84$ and $5.08$ across the four penalties, so more than five of the ten
directions are absorbed. Fitting $62$ penalized coefficients costs far more of the
residual space than the classical count admits. The textbook $\chi^2_{G-2}$ reference is
therefore too diffuse, and gives
$p$-values from $7.5\times10^{-7}$ to $0.115$ across the same grid, failing to reject in
one cell of the twelve. Comparing our correction against that weaker version would
flatter it, so we do not.

(b) The corrected test returns $p \approx 0.03$ at the cross-validated penalty with ten
groups, and values between $0.014$ and $0.678$ across the grid. The evidence for misfit
is therefore real but modest, and it is weaker than the uncorrected analysis suggested
by more than three orders of magnitude. We deliberately do not claim that the corrected
test exonerates the model. It does not. What it does is replace an uninterpretable
number with an interpretable one.

(c) The corrected $p$-value varies with the number of groups, from $0.026$ at ten groups
to $0.172$ at twenty. Sensitivity to grouping is inherited from the Hosmer--Lemeshow
family and our correction does not remove it; \citet{nattino2020} address it directly by
replacing the grouping with a fixed partition of the probability scale, which is
complementary to the correction developed here. We report the whole grid rather than a
single cell, and we would encourage the same practice generally.

\subsection{Reanalysis of a second dataset}\label{sec:mvf}

A reader is entitled to ask whether one dataset carries this much weight. We therefore
ran the identical grid on \texttt{GlaucomaMVF} (package \texttt{ipred}), a second series
of $n = 170$ eyes described by $p = 63$ of the same morphometric variables, measured with
the same instrument, so that $p/n = 0.371$. It is not an independent replication --- it
is the same clinical question recorded with the same protocol --- but it is a different
sample, and the maximum likelihood estimate does not exist there either.

The outcome is not the same, and that is the point. On \texttt{GlaucomaM} the correction
turns an overwhelming verdict into a borderline one and does not exonerate the model. On
\texttt{GlaucomaMVF} it produces a clean reversal: the uncorrected test rejects in all
twelve configurations on the EDGE basis, and the corrected test rejects in \emph{none} of
them on either basis, returning $p$ between $0.224$ and $0.998$ on the decile basis and
between $0.298$ and $0.874$ on the EDGE basis. Repeating the cross-validated cell under
five seeds leaves those conclusions unchanged. The full grid is given in the
supplementary material.

Two different verdicts from one mechanism is what we would expect and what we want. The
correction does not manufacture acceptance; it removes a displacement whose size depends
on the penalty and the design, and what is left over is whatever misfit the data actually
contain. On one dataset that residue is real but modest; on the other there is none to
find. An analyst who applied the uncorrected test to both would have concluded that both
models were badly misspecified, and would have been wrong about at least one of them.

\subsection{Implications for clinical practice}\label{sec:clinicalchange}

The two analyses support different actions.

Under the uncorrected analysis the model is rejected outright. The natural responses are
to abandon it, to rebuild it with different predictors, or to collect more data --- and
in a screening application where the alternative is subjective expert assessment of the
same optic disc images \citep{coan2023}, discarding a model with an area under the curve
of $0.905$ has a cost measured in missed early disease.

Under the corrected analysis the model is not rejected outright, and the residual
evidence of misfit has a recognizable shape: the predicted probabilities are too flat,
not mis-ordered. Discrimination is unaffected by that, so the model remains usable for
ranking patients by risk --- which is what a screening triage actually requires. Where
the model should not be used unaltered is where the predicted probability is read as a
number, for instance when it is compared against a fixed threshold to decide referral.
The remedy there is recalibration rather than reconstruction, and the direction of the
required correction is known from the fit itself.

This distinction has a reassuring consequence that is worth stating explicitly, because
it is not obvious. Penalization compresses risks towards the average, which in the
terminology of \citet{vancalster2015} is underfitting rather than overfitting. Of the
four kinds of miscalibration they examine, underfitting is the only one that never made
a model clinically harmful in the sense of being worse than a default strategy. A
penalized model that has been wrongly condemned by an uncorrected goodness-of-fit test
is therefore likely to be one whose miscalibration was, in decision-analytic terms, the
benign kind. The cost of the invalid test is not that it protects patients too
zealously; it is that it discards usable models for the wrong reason.

\subsection{Limitations of the application}\label{sec:applimit}

These are development data, and every quantity we report is apparent rather than
validated. Our purpose is a like-for-like comparison of two tests on the same fit, which
does not require external validation; but the calibration slope of $2.43$ should not be
read as an estimate of what the model would do in a new population, and a model intended
for clinical use would require external validation and, on this evidence, recalibration
before deployment \citep{vancalster2020}.

\section{Software}\label{sec:software}

The EDGE basis of \citet{ebrahim2026edge} is implemented as \texttt{edge.gof()} in the R
package \texttt{ebrahim.gof} (version 2.4.0) on the Comprehensive R Archive Network, and
the EDGE statistic computed by the code used here agrees with that implementation exactly
on three datasets. The shrinkage correction itself is not yet in a released version of
that package. A reference implementation is included in the replication archive as
\texttt{shrink.gof()}, which returns \texttt{SC.HL} and \texttt{SC.EDGE} for a fitted
ridge model; it depends only on base \textsf{R}, so that reproducing the results here
does not depend on any package version. On the glaucoma data of
Section~\ref{sec:app} it returns the entries of Table~\ref{tab:glaucoma} to within
bootstrap error. The scale identity $\lambda = n\lambda_{\mathrm{g}}$ against
\citet{friedman2010} was checked numerically to $1.6\times10^{-9}$.

\section{Discussion}\label{sec:disc}

The general lesson is the separation of Section~\ref{sec:twonulls}. Whenever a fitting
procedure deliberately biases the estimator --- ridge, lasso, elastic net, Firth's
correction, any Bayesian prior with a non-negligible influence --- the structural null
and the calibration null stop being the same hypothesis, and a diagnostic built for the
first will read the deliberate bias as evidence against the model. Grouped goodness-of-fit tests are the
case treated here because they are the ones clinical prediction papers actually use, but
the mechanism is not specific to them: any diagnostic whose reference distribution is
centred by an unbiasedness argument inherits the same problem.

Three things are deferred and we name them as open rather than solved.

(a) Our theory is first-order, treats $\lambda$ as non-random and requires
$\lambda = o(n)$, and holds the dimension fixed. The empirical behaviour we rely on most
--- validity at $p/n = 0.25$, and the glaucoma analysis itself at $p/n = 0.32$ --- comes
from prepivoting, for which we have no analytic account in the proportional regime where
$p/n \to \kappa \in (0,1)$. The results of \citet{sur2019} and \citet{candes2020} suggest
that such an account exists and would be worth having.

(b) We treat the ridge penalty, and Remark~\ref{rem:smooth} extends the argument to any
smooth penalty. The lasso is genuinely outside it, but the reason is sharper than
non-differentiability. For the lasso the stationarity condition reads
$X^{\top}(\by - \hat\bpi) = \lambda\hat{\boldsymbol{s}}$ for a subgradient
$\hat{\boldsymbol{s}}$, so the same substitution returns the debiased lasso rather than
nothing at all. The obstruction is that $\hat{\boldsymbol{s}}$ is not a deterministic
function of the fit on the active set, and $F^{-1}$ must be replaced by a nodewise
approximation; the displacement $\boldsymbol{\mu}$ therefore becomes random, and the
exact cancellation of Proposition~\ref{prop:cancel} no longer holds identically. Whether
it holds to sufficient order is open, and we do not claim it.

(c) We fix $\lambda$ throughout. In practice $\lambda$ is chosen by cross-validation on
the same data, and the resulting selection effect is not accounted for by our procedure
or, so far as we know, by any other.

We are also explicit about what the test is. It is a test of calibration along the fitted
index. It is not an omnibus test of the logistic structure, and Section~\ref{sec:scope}
quantifies exactly how much of a departure it can be expected to see. We regard that
quantification as part of the contribution: a diagnostic whose scope is known is more
useful than one whose scope is assumed.

The last word belongs to the eyes. A model built to detect glaucoma early, from
measurements a machine already takes, was reported by the standard check as badly
misspecified --- at every penalty and at every grouping we tried. On that reading
it would have been rebuilt or abandoned, and in a disease whose damage cannot be
reversed, the cost of abandoning a usable screening model is counted in sight that was
not saved. The corrected analysis says something different and more useful: the fitted
probabilities are too flat, not mis-ordered; the ranking on which a screening triage
depends is sound; and what the model needs before anyone reads its probabilities as
numbers is recalibration, not reconstruction. That is a smaller and more actionable
conclusion than the one the uncorrected test appeared to license, and it is the one the
data actually support. We think it is also the more common situation. Wherever a
prediction model has been penalized because it had to be, and then judged by a test that
assumed it had not been, the verdict on record may be a verdict about the estimator
rather than about the model.

\section*{Data availability}
The \texttt{GlaucomaM} data are distributed in the R package \texttt{TH.data} on the
Comprehensive R Archive Network. The predictors are standardized to zero mean and unit
variance before fitting; the data are otherwise unmodified. The second dataset,
\texttt{GlaucomaMVF}, is distributed in the R package \texttt{ipred}. All code, seeds and
intermediate results required to reproduce every number and figure in this paper, together
with the per-replication $p$-values described in the supplementary material, are available
at \url{https://github.com/ebrahimkhaled/gof-penalized-paper} and archived at
\url{https://doi.org/10.5281/zenodo.21903202}.

\section*{Acknowledgements}
All praise and thanks are due to Allah, the Almighty, for every blessing that made this
work possible.

I thank the Department of Applied Statistics, Alexandria University.
I am grateful to Torsten Hothorn and the maintainers of the \texttt{TH.data} and
\texttt{ipred} packages for making the glaucoma data publicly available, and to the
maintainers of \texttt{glmnet} and \texttt{GRPtests}.

In line with the journal's policy on artificial intelligence tools: an AI assistant
(Claude, Anthropic) was used to draft and check the simulation code, prepare figures,
and edit the prose. All derivations, numerical results and conclusions were verified by
the author, who takes full responsibility for the content.

\section*{Funding}
No funding was received for this work.

\section*{Conflict of interest}
None declared.



\begin{thebibliography}{99}

\bibitem[Beran(1987)]{beran1987}
Beran, R. (1987).
Prepivoting to reduce level error of confidence sets.
\emph{Biometrika}, \textbf{74}(3), 457--468.

\bibitem[Beran(1988)]{beran1988}
Beran, R. (1988).
Prepivoting test statistics: a bootstrap view of asymptotic refinements.
\emph{Journal of the American Statistical Association}, \textbf{83}(403), 687--697.

\bibitem[Cand\`es and Sur(2020)]{candes2020}
Cand\`es, E.J. and Sur, P. (2020).
The phase transition for the existence of the maximum likelihood estimate in
high-dimensional logistic regression.
\emph{The Annals of Statistics}, \textbf{48}(1), 27--42.

\bibitem[Chernoff and Lehmann(1954)]{chernoff1954}
Chernoff, H. and Lehmann, E.L. (1954).
The use of maximum likelihood estimates in $\chi^2$ tests for goodness of fit.
\emph{The Annals of Mathematical Statistics}, \textbf{25}(3), 579--586.

\bibitem[Coan et~al.(2023)]{coan2023}
Coan, L.J., Williams, B.M., Krishna Adithya, V. et~al. (2023).
Automatic detection of glaucoma via fundus imaging and artificial intelligence: a review.
\emph{Survey of Ophthalmology}, \textbf{68}(1), 17--41.

\bibitem[Collins et~al.(2015)]{collins2015}
Collins, G.S., Reitsma, J.B., Altman, D.G. and Moons, K.G.M. (2015).
Transparent reporting of a multivariable prediction model for individual prognosis or
diagnosis (TRIPOD): the TRIPOD statement.
\emph{BMJ}, \textbf{350}, g7594.

\bibitem[Copas(1989)]{copas1989}
Copas, J.B. (1989).
Unweighted sum of squares test for proportions.
\emph{Journal of the Royal Statistical Society, Series C}, \textbf{38}(1), 71--80.

\bibitem[Dawid(1982)]{dawid1982}
Dawid, A.P. (1982).
The well-calibrated Bayesian.
\emph{Journal of the American Statistical Association}, \textbf{77}(379), 605--613.

\bibitem[Ebrahim and El-Kotory(2026)]{ebrahim2026edge}
Ebrahim, E.K. and El-Kotory, A. (2026).
EDGE: a closed-form directed test for the calibration of probabilistic binary classifiers.
Submitted to \emph{Advances in Data Analysis and Classification}. Preprint,
\url{https://doi.org/10.5281/zenodo.21902134}.

\bibitem[Firth(1993)]{firth1993}
Firth, D. (1993).
Bias reduction of maximum likelihood estimates.
\emph{Biometrika}, \textbf{80}(1), 27--38.

\bibitem[Friedman et~al.(2010)]{friedman2010}
Friedman, J., Hastie, T. and Tibshirani, R. (2010).
Regularization paths for generalized linear models via coordinate descent.
\emph{Journal of Statistical Software}, \textbf{33}(1), 1--22.

\bibitem[Gulati et~al.(2022)]{gulati2022}
Gulati, G., Upshaw, J., Wessler, B.S. et~al. (2022).
Generalizability of cardiovascular disease clinical prediction models: 158 independent
external validations of 104 unique models.
\emph{Circulation: Cardiovascular Quality and Outcomes}, \textbf{15}(4), e008487.

\bibitem[Guo and Chen(2016)]{guo2016}
Guo, B. and Chen, S.X. (2016).
Tests for high dimensional generalized linear models.
\emph{Journal of the Royal Statistical Society, Series B}, \textbf{78}(5), 1079--1102.

\bibitem[Hosmer and Lemeshow(1980)]{hosmer1980}
Hosmer, D.W. and Lemeshow, S. (1980).
Goodness of fit tests for the multiple logistic regression model.
\emph{Communications in Statistics -- Theory and Methods}, \textbf{9}(10), 1043--1069.

\bibitem[Hosmer et~al.(1997)]{hosmer1997}
Hosmer, D.W., Hosmer, T., le~Cessie, S. and Lemeshow, S. (1997).
A comparison of goodness-of-fit tests for the logistic regression model.
\emph{Statistics in Medicine}, \textbf{16}(9), 965--980.

\bibitem[Jankov\'a et~al.(2020)]{jankova2020}
Jankov\'a, J., Shah, R.D., B\"uhlmann, P. and Samworth, R.J. (2020).
Goodness-of-fit testing in high dimensional generalized linear models.
\emph{Journal of the Royal Statistical Society, Series B}, \textbf{82}(3), 773--795.

\bibitem[Javanmard and Mehrabi(2024)]{javanmard2024}
Javanmard, A. and Mehrabi, M. (2024).
GRASP: a goodness-of-fit test for classification learning.
\emph{Journal of the Royal Statistical Society, Series B}, \textbf{86}(1), 215--245.

\bibitem[le~Cessie and van~Houwelingen(1991)]{lecessie1991}
le~Cessie, S. and van~Houwelingen, J.C. (1991).
A goodness-of-fit test for binary regression models, based on smoothing methods.
\emph{Biometrics}, \textbf{47}(4), 1267--1282.

\bibitem[le~Cessie and van~Houwelingen(1992)]{lecessie1992}
le~Cessie, S. and van~Houwelingen, J.C. (1992).
Ridge estimators in logistic regression.
\emph{Journal of the Royal Statistical Society, Series C}, \textbf{41}(1), 191--201.

\bibitem[Moore and Spruill(1975)]{moore1975}
Moore, D.S. and Spruill, M.C. (1975).
Unified large-sample theory of general chi-squared statistics for tests of fit.
\emph{The Annals of Statistics}, \textbf{3}(3), 599--616.

\bibitem[Nattino et~al.(2020)]{nattino2020}
Nattino, G., Pennell, M.L. and Lemeshow, S. (2020).
Assessing the goodness of fit of logistic regression models in large samples: a
modification of the Hosmer--Lemeshow test.
\emph{Biometrics}, \textbf{76}(2), 549--560.

\bibitem[Pavlou et~al.(2024)]{pavlou2024}
Pavlou, M., Omar, R.Z. and Ambler, G. (2024).
Penalized regression methods with modified cross-validation and bootstrap tuning produce
better prediction models.
\emph{Biometrical Journal}, \textbf{66}(5), e2300245.

\bibitem[Puhr et~al.(2017)]{puhr2017}
Puhr, R., Heinze, G., Nold, M., Lusa, L. and Geroldinger, A. (2017).
Firth's logistic regression with rare events: accurate effect estimates and
predictions?
\emph{Statistics in Medicine}, \textbf{36}(14), 2302--2317.

\bibitem[Riley et~al.(2021)]{riley2021}
Riley, R.D., Snell, K.I.E., Martin, G.P. et~al. (2021).
Penalization and shrinkage methods produced unreliable clinical prediction models
especially when sample size was small.
\emph{Journal of Clinical Epidemiology}, \textbf{132}, 88--96.

\bibitem[Shah and B\"uhlmann(2018)]{shah2018}
Shah, R.D. and B\"uhlmann, P. (2018).
Goodness-of-fit tests for high dimensional linear models.
\emph{Journal of the Royal Statistical Society, Series B}, \textbf{80}(1), 113--135.

\bibitem[Sur and Cand\`es(2019)]{sur2019}
Sur, P. and Cand\`es, E.J. (2019).
A modern maximum-likelihood theory for high-dimensional logistic regression.
\emph{Proceedings of the National Academy of Sciences}, \textbf{116}(29), 14516--14525.

\bibitem[Sur et~al.(2019)]{surchen2019}
Sur, P., Chen, Y. and Cand\`es, E.J. (2019).
The likelihood ratio test in high-dimensional logistic regression is asymptotically a
rescaled chi-square.
\emph{Probability Theory and Related Fields}, \textbf{175}(1--2), 487--558.

\bibitem[Van Calster and Vickers(2015)]{vancalster2015}
Van Calster, B. and Vickers, A.J. (2015).
Calibration of risk prediction models: impact on decision-analytic performance.
\emph{Medical Decision Making}, \textbf{35}(2), 162--169.

\bibitem[Van Calster et~al.(2016)]{vancalster2016}
Van Calster, B., Nieboer, D., Vergouwe, Y., De~Cock, B., Pencina, M.J. and
Steyerberg, E.W. (2016).
A calibration hierarchy for risk models was defined: from utopia to empirical data.
\emph{Journal of Clinical Epidemiology}, \textbf{74}, 167--176.

\bibitem[Van Calster et~al.(2019)]{vancalster2019}
Van Calster, B., McLernon, D.J., van~Smeden, M. et~al. (2019).
Calibration: the Achilles heel of predictive analytics.
\emph{BMC Medicine}, \textbf{17}, 230.

\bibitem[Van Calster et~al.(2020)]{vancalster2020}
Van Calster, B., van~Smeden, M., De~Cock, B. and Steyerberg, E.W. (2020).
Regression shrinkage methods for clinical prediction models do not guarantee improved
performance: simulation study.
\emph{Statistical Methods in Medical Research}, \textbf{29}(11), 3166--3178.

\bibitem[van~de~Geer et~al.(2014)]{vandegeer2014}
van~de~Geer, S., B\"uhlmann, P., Ritov, Y. and Dezeure, R. (2014).
On asymptotically optimal confidence regions and tests for high-dimensional models.
\emph{The Annals of Statistics}, \textbf{42}(3), 1166--1202.

\bibitem[Zhang and Zhang(2014)]{zhangzhang2014}
Zhang, C.-H. and Zhang, S.S. (2014).
Confidence intervals for low dimensional parameters in high dimensional linear
models.
\emph{Journal of the Royal Statistical Society, Series B}, \textbf{76}(1), 217--242.

\end{thebibliography}
\end{document}